\documentclass[a4paper,10pt]{IEEEtran}

\usepackage{amsmath}
\usepackage{amssymb}
\usepackage{amsfonts}
\usepackage{graphicx}

\usepackage{amsthm}
\newtheorem{proposition}{Proposition}
\newtheorem*{conjecture*}{Conjecture}

\usepackage{algorithm}
\usepackage{algpseudocode}

\usepackage{xcolor}

\usepackage{physics}

\usepackage{setspace}

\usepackage{cite}

\DeclareMathOperator*{\argmin}{arg\,min}

\algblockdefx{Repeat}{EndRepeat}[1]{\textbf{repeat} #1}{\textbf{end repeat}}

\title{Online Convex Optimization of Programmable Quantum Computers to Simulate Time-Varying Quantum Channels}

\author{
\IEEEauthorblockN{Hari Hara Suthan Chittoor,~\IEEEmembership{Member,~IEEE}, Osvaldo Simeone,~\IEEEmembership{Fellow,~IEEE},\\ Leonardo Banchi and Stefano Pirandola\vspace{-0.9cm}}}

\begin{document}

\maketitle

\thispagestyle{empty}	

\pagestyle{empty}


\let\thefootnote\relax\footnotetext
{
Hari Hara Suthan Chittoor and Osvaldo Simeone are with King’s Communications, Learning, and Information
Processing (KCLIP) lab at the Department of Engineering of Kings College
London, UK (emails: hari.hara@kcl.ac.uk, osvaldo.simeone@kcl.ac.uk). Their work has been supported by the European Research Council (ERC) under the European Union's Horizon 2020 Research and Innovation Programme (Grant Agreement No. 725731), and Osvaldo Simeone has also been supported by an Open Fellowship of the EPSRC (EP/W024101/1). For the purpose of open access, the author has applied a Creative Commons Attribution (CC BY) licence to any Author Accepted Manuscript version arising. The authors acknowledge use of the research computing facility at King’s College London, Rosalind (https://rosalind.kcl.ac.uk).

Leonardo Banchi is with the Department of Physics and Astronomy, University of Florence \& INFN sezione di Firenze, via G. Sansone 1, I-50019 Sesto Fiorentino (FI), Italy (email: leonardo.banchi@unifi.it). His work is supported by the U.S. Department of Energy, Office of Science, National
Quantum Information Science Research Centers, Superconducting Quantum Materials and Systems Center
(SQMS) under the contract No. DE-AC02-07CH11359.

Stefano Pirandola is with the Department of Computer Science, University of York, York YO10 5GH, UK (email: stefano.pirandola@york.ac.uk).


}

\begin{abstract}

Simulating quantum channels is a fundamental primitive in quantum computing, since quantum channels define general (trace-preserving) quantum operations. An arbitrary quantum channel cannot be exactly simulated using a finite-dimensional programmable quantum processor, making it important to develop optimal approximate simulation techniques. In this paper, we study the challenging setting in which the channel to be simulated varies adversarially with time. We propose the use of matrix exponentiated gradient descent (MEGD), an online convex optimization method, and analytically show that it achieves a sublinear regret in time. Through experiments, we validate the main results for time-varying dephasing channels using a programmable generalized teleportation processor.


\end{abstract}


\vspace{-0.2cm}

\section{Introduction}

A quantum computer can be programmed to carry out a given functionality in different ways, including the direct engineering of pulse sequences \cite{pulse_techniques_for_quantum_2016}, the design of parametric quantum circuits via quantum machine learning \cite{Book_Machine_learning_quantum_computers_Schuld_2021, Book_Osvaldo_quantum_machine_learning_for_engineers}, the use of adaptive measurements on cluster states \cite{Fault_tolerant_quantum_computation_cluster_states_Nielsen_2005}, and the optimization of a program state operating on a fixed quantum processor. A fundamental result derived in \cite{Nielsen_progm_Q_gate_arrays} states there is no universal programmable quantum processor that operates with finite-dimensional program states. Since a quantum processor is universal if it can implement any quantum operation, this conclusion implies that the exact simulation of an arbitrary quantum channel on a single programmable quantum processor is impossible. This, in turn, highlights the importance of developing tools for the optimization of quantum programs.


 Reference \cite{banchi_convex_opt_prog_q_comp_2020} addressed the problem of approximately simulating a quantum channel using a finite-dimensional program state. The authors proved that the error between the target channel and simulated channel, as measured by the diamond distance, as well as other related metrics,  is convex in the space of program states. Specifically, the optimal program state can be calculated using semidefinite programming.  In this paper, we study the more challenging setting illustrated in Fig. \ref{fig: programmable quantum computation}, in which the channel to be simulated varies over time. We adopt a worst-case formulation in which channel variation are arbitrary, and chosen by ``nature" in a possibly adversarial way.


\begin{figure}[t]
    \centering
    \includegraphics[height=1.8in]{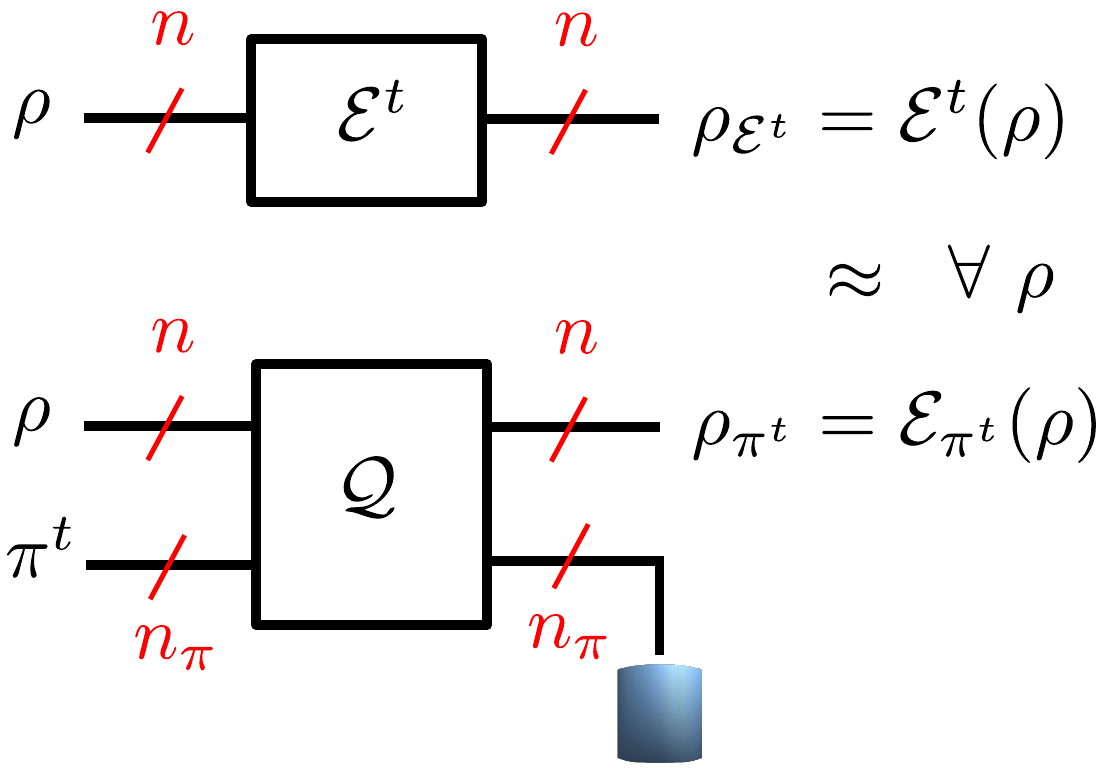} \vspace{-0.2cm}
    \caption{Time-varying quantum channel $\mathcal{E}^t$ (top) and its simulation $\mathcal{E}_{\pi^{t}}$ via a programmable quantum processor $\mathcal{Q}$ controlled by the time-varying program state $\pi^t$ (bottom).}
    \label{fig: programmable quantum computation}
    \vspace{-0.6cm}
\end{figure}



To study this setting, we propose to adopt the framework of online convex optimization \cite{Orabona_online_learning_notes}, which provides tools to 
track the optimal solution of time-varying convex problems. We specifically develop and analyze an online mirror descent algorithm over the space of positive definite matrices, yielding a matrix 
exponentiated gradient descent (MEGD) \cite{MEGD_tsuda_2005}. MEGD was previously used in the related, but distinct, problem of quantum state tracking \cite{MEGD_online_learning_of_quantum_state_NeurIPS_2018,MEGD_online_quantum_state_tracking_in_noisy_env}.

We prove that the regret of MEGD with respect to an optimized fixed program state is sublinear in time. We conduct experiments by adopting the generalized teleportation processor (GTP) as the programmable quantum processor. GTP can simulate exactly the class of teleportation-covariant channels, modeling Pauli and erasure channels \cite{Book_quantum_computing_quantum_information_Nielsen_chuang_2010}, and is operated here in an adversarial setting with time varying dephasing channels. Numerical results validate the analysis.

\noindent \textbf{Notations and Definitions:} 
For any non-negative integer $K$, $[K]$ represents the set $\{1,2,\cdots,K\}$. The number of elements in a set $\mathcal{A}$ is written as $|\mathcal{A}|$. Given two sets $\mathcal{A}$ and $\mathcal{B}$, we write $\mathcal{A} \setminus \mathcal{B} =\{a : a \in \mathcal{A} \text{ and } a \notin \mathcal{B}\}$. The symbol $\forall$ represents for all.
The Kronecker product is denoted as $\otimes$; $I$ represents the identity matrix, with dimensions clear from the context; $M^{\dagger}$ represents the complex conjugate transpose of the matrix $M$; and $\mathrm{tr}(M)$ represents trace of the matrix $M$. We adopt standard notations for quantum states, computational basis, and quantum gates \cite{Book_quantum_computing_quantum_information_Nielsen_chuang_2010}.

\section{Problem Formulation}
\label{section problem formulation}

In this section, we first review some background material and then describe the setting and problem of interest.

\subsection{Preliminaries}

Throughout this paper, we use the standard Dirac notation (see, e.g., \cite{Book_quantum_computing_quantum_information_Nielsen_chuang_2010}).
Given $n$ qubits, we let $\mathcal{D}(\mathcal{H})$ denote the space of all density matrices, i.e., positive semidefinite (PSD) matrices with unit trace, defined on the Hilbert space $\mathcal{H}$ of dimension $2^n$. Any $2^n \times 2^n$ Hermitian matrix $A$ can be written in terms of its eigendecomposition $A = \sum_{i} \lambda_i |v_i\rangle \langle v_i|$, where eigenvalues $
\{\lambda_i
\}$ are real and the set $\{|v_i\rangle \}$ consists of a basis of orthonormal vectors for the $2^n$-dimensional Hilbert space. Furthermore, the square root of a PSD matrix $A$ is defined as $\sqrt{A} = \sum_{i} \sqrt{\lambda_i} |v_i\rangle \langle v_i|$. More generally, a function $f(A)$ of PSD matrix $A$ is defined as $f(A) = \sum_{i} f(\lambda_i) |v_i\rangle \langle v_i|$.


A quantum channel $\mathcal{E}$ is a completely positive trace preserving (CPTP) linear map that takes a density matrix $\rho \in \mathcal{D}(\mathcal{H})$ as input to produce a density matrix $\mathcal{E}(\rho) \in \mathcal{D}(\mathcal{H}^{'})$ in a possibly distinct Hilbert space $\mathcal{H}^{'}$ of dimension $2^{n^{'}}$ for an integer $n^{'}$.  Furthermore, given a system of $2n$ qubits, we denote as $I\otimes \mathcal{E}$, where $I$ is the $2^n \times 2^n$ identity matrix, the channel that acts trivially on the first $n$ qubits and applies channel $\mathcal{E}$ to the last $n$ qubits.
    
A quantum channel $\mathcal{E}$ can be equivalently described by the $2^{2n} \times 2^{2n}$ PSD matrix obtained as the output of channel $I \otimes \mathcal{E}$ applied to  a system of $2n$ qubits in the Bell state $|\Phi^+ \rangle = 1/\sqrt{2^n}\sum_{i=0}^{2^{n} -1} |i\rangle \otimes |i\rangle $. This matrix, known as Choi matrix of the quantum channel $\mathcal{E}$, is hence defined as
\begin{equation}
    C_{\mathcal{E}} = (I \otimes \mathcal{E}) |\Phi^+ \rangle \langle \Phi^+ | = \frac{1}{2^n} \sum\limits_{i=0}^{2^n -1} \sum\limits_{j=0}^{2^n -1} |i\rangle \langle j| \otimes \mathcal{E}(|i\rangle \langle j|).    
\end{equation}



\vspace{-0.2cm}

\subsection{Setting}

\vspace{-0.1cm}

As shown in Fig. \ref{fig: programmable quantum computation}, we study the problem of simulating a time-varying quantum channel $\mathcal{E}^t$ operating on $n$ qubits using a programmable quantum processor $\mathcal{Q}$, where $t$ is a discrete time index $t =1,2,\ldots$ Specifically, the top part of Fig. \ref{fig: choi matrix illustation} illustrates the Choi matrix $C_{\mathcal{E}^t}$ of the quantum channel $\mathcal{E}^t$.
As depicted in the bottom part of Fig. \ref{fig: choi matrix illustation}, the programmable quantum processor is a fixed CPTP map $\mathcal{Q}$ operating on a register of $n + n_{\pi}$ qubits. Examples of quantum processors $\mathcal{Q}$ include generalized teleportation processor \cite{Bennett_teleportation_fundamental_paper_1993,Advances_in_teleportatin_pirandola_2015,channel_simulaiton_in_quantum_metrology_2018} and the port-based teleportation processor \cite{characterising_PBT_as_univ_simulator_of_Q_channels_stefano_2021,PBT_Ishizaka_2009}. The quantum processor $\mathcal{Q}$ is ``programmable" via a time-varying program state $\pi^t \in \mathcal{D}(\mathcal{H}_{\pi})$, where $\mathcal{H}_{\pi}$ is a $2^{n_{\pi}}$-dimensional Hilbert space.


\begin{figure}[t]
    \centering
    \includegraphics[height=2.2in]{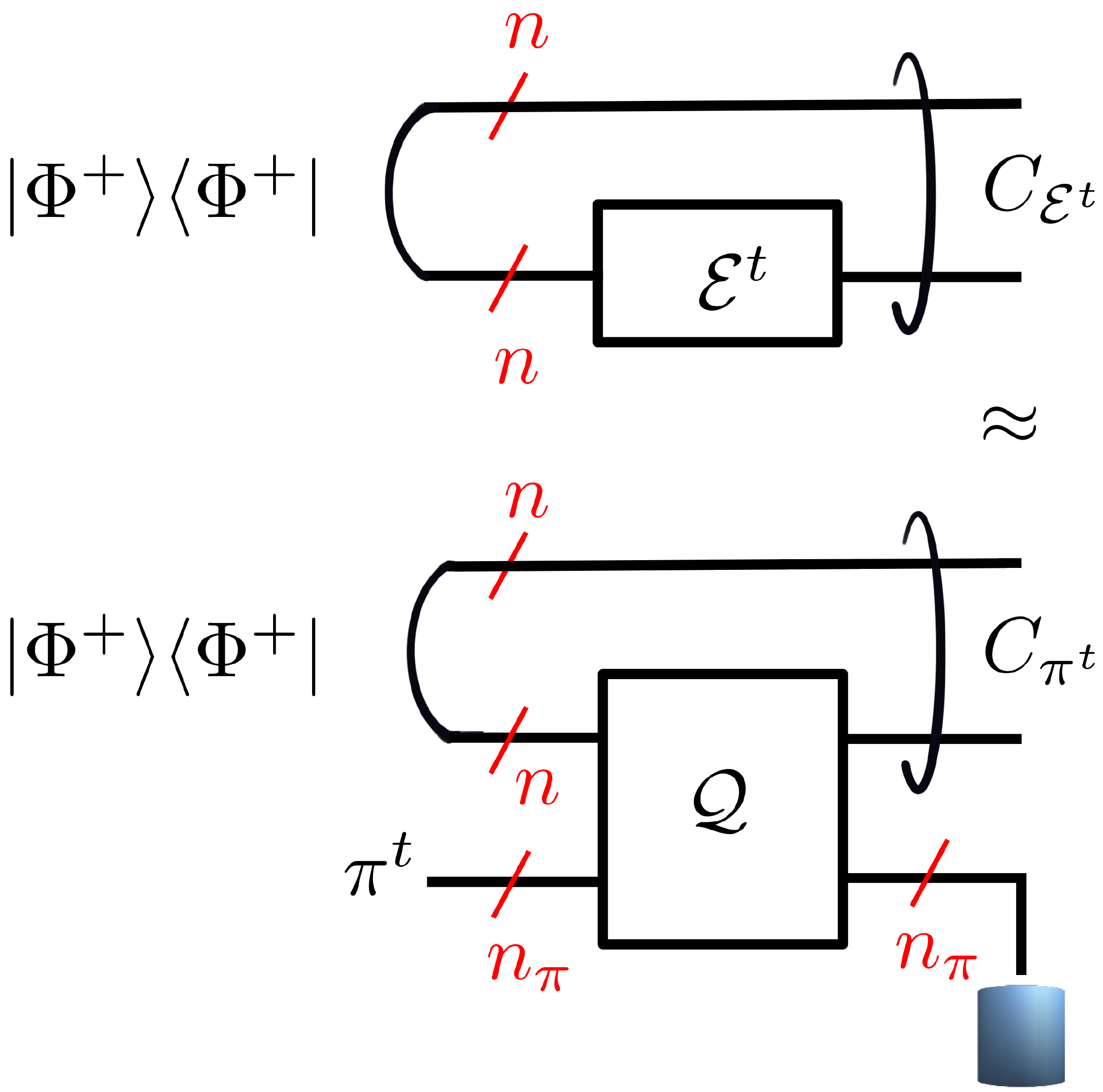} \vspace{-0.2cm}
    \caption{Illustration of the Choi matrix $C_{\mathcal{E}^t}$ of a quantum channel $\mathcal{E}^t$ (top) and of the Choi matrix $C_{\mathcal{E}_{\pi^t}} = C_{\pi^t}$ of its simulation $\mathcal{E}_{\pi^t}$ via a programmable quantum processor $\mathcal{Q}$ controlled by the program state $\pi^t$ (bottom).}
    \label{fig: choi matrix illustation}
    \vspace{-0.5cm}
\end{figure}


The register of $n+n_{\pi}$ qubits on which processor $\mathcal{Q}$ operates is initially in state $\rho \otimes \pi^t$, where $\rho \in \mathcal{D}(\mathcal{H})$ is the input density and $\pi^t$ the program state. After the application of processor $\mathcal{Q}$, we wish for the reduced state of a subset of $n$ qubits within the register of $n+n_{\pi}$ qubits to approximate the output state
\begin{equation}
    \rho_{\mathcal{E}^t} = \mathcal{E}^t(\rho)
\end{equation}
of channel $\mathcal{E}^t$ for any possible input state $\rho$. We refer to the mentioned subset of $n$ qubits as forming the output subregister. This may differ from the subregister consisting of the first $n$ qubits initially in the input state $\rho$ (see Sec. \ref{section experiments} for an example).

To simulate the channel $\mathcal{E}^t$, we optimize the sequence of program states $\pi^t$ sequentially over time $t \in [T]$. As we will detail in Section \ref{subsection problem definition}, at each time $t$, the optimizer has access to information about the quality of the approximation of the channels $\mathcal{E}^{\tau}$ obtained with program states $\pi^{\tau}$ at previous times $\tau \in [t-1]$.

Accordingly, given a program state $\pi^t \in \mathcal{D}(\mathcal{H}_{\pi})$, the programmable quantum processor $\mathcal{Q}$ implements the channel $\mathcal{E}_{\pi^{t}}$ defined by the mapping
\begin{equation}
\label{eq output of p.q.c}
    \mathcal{E}_{\pi^t} (\rho) = \mathrm{tr}_{\pi} (\mathcal{Q}(\rho \otimes \pi^t)),
\end{equation}
where $\mathrm{tr}_{\pi}(\cdot)$ is the partial trace over the $n_\pi$ qubits outside the output register. By (\ref{eq output of p.q.c}), as illustrated in Fig. \ref{fig: programmable quantum computation}, these $n_{\pi}$ qubits are discarded after the application of the operation $\mathcal{Q}$.

The simulation error at time $t$ is a measure of the difference between the channel $\mathcal{E}^t$ and the simulated channel $\mathcal{E}_{\pi^{t}}$. As illustrated in Fig. \ref{fig: choi matrix illustation}, this error can be measured by comparing the corresponding Choi matrices $C_{\mathcal{E}^t}$ and $C_{\mathcal{E}_{\pi^t}} = C_{\pi^t}$. Following \cite{banchi_convex_opt_prog_q_comp_2020}, the loss is specifically quantified by either the trace distance
\begin{equation}
    \label{eq trace distrance}
    \ell_1 (\mathcal{E}^t,\pi^t) = \frac{1}{2} \| C_{\mathcal{E}^t} - C_{\pi^t}\|_{\mathrm{tr}},
\end{equation}
where $\|O \|_{\mathrm{tr}} = \mathrm{tr}(\sqrt{O^{\dagger}O})$ is the trace of positive square root of matrix $O^{\dagger}O$; or alternatively, by the complement of the squared fidelity
\begin{equation}
    \label{eq infidelity metric}
    \ell_F (\mathcal{E}^t, \pi^t) = 1- \left(\mathrm{tr}\left( \sqrt{ \sqrt{C_{\mathcal{E}^t}} C_{\pi^t} \sqrt{ C_{\mathcal{E}^t} }}\right) \right)^{2}.
\end{equation}
We write $\ell(\mathcal{E}^t, \pi^t)$ to denote either loss (\ref{eq trace distrance}) or (\ref{eq infidelity metric}).



In \cite{banchi_convex_opt_prog_q_comp_2020}, the authors studied the problem 
\begin{equation}
\label{eq optimization of loss function}
    \min_{\pi \in \mathcal{D}(\mathcal{H}_{\pi})} \ell(\mathcal{E}, \pi)
\end{equation}
of optimizing the program state $\pi$ given a fixed quantum channel $\mathcal{E}$. Reference \cite{banchi_convex_opt_prog_q_comp_2020} proved that the optimization problem (\ref{eq optimization of loss function}) is convex over the program states $\pi$. Note that the work \cite{banchi_convex_opt_prog_q_comp_2020} considered also the diamond distance. which can be upper bounded via the loss functions (\ref{eq trace distrance}) and (\ref{eq infidelity metric}).


\subsection{Problem Definition}
\label{subsection problem definition}

Unlike \cite{banchi_convex_opt_prog_q_comp_2020}, we consider the problem of online optimization of the program state $\pi^t$ for time-varying channels $\mathcal{E}^t$ over time $t \in [T]$. We formulate the problem in an adversarial setting to obtain worst-case performance guarantees in terms of the possible sequence of channels $\mathcal{E}^{[T]} = \mathcal{E}^{1}, \mathcal{E}^{2}, \cdots, \mathcal{E}^{T} $.
Accordingly, at every time $t \in [T]$, the optimizer produces a program state $\pi^t$. Then, a quantum channel $\mathcal{E}^t$ is selected in an arbitrary way by ``nature", and the optimizer pays the loss $\ell(\mathcal{E}^t, \pi^t)$, which measures how poorly the simulated channel (\ref{eq output of p.q.c}) obtained with program state $\pi^t$ approximates channel $\mathcal{E}^t$. For every time $t\in [T]$, the optimizer produces an updated program state $\pi^{t+1}$ assuming access to a subgradient of the current loss, i.e.,  $g(\mathcal{E}^t, \pi^t) =\nabla_{\pi} \ell(\mathcal{E}^t, \pi^t)$.

Let us define the policy followed by the optimizer as the function
\begin{equation}
\label{eq policy}
    \pi^{t+1} = f^t \left(\pi^t,  g(\mathcal{E}^t, \pi^t) \right),
\end{equation}
which maps the current program state $\pi^t$ and subgradient $g(\mathcal{E}^t, \pi^t)$ to the next program state $\pi^{t+1}$. The goal is to design a sequence of functions $f^{[T]}(\cdot) = f^{1}(\cdot), f^{2}(\cdot), \cdots, f^{T}(\cdot)$ that performs well when compared to a fixed reference program state $\pi^*$ optimized based on knowledge of the sequence $\mathcal{E}^{[T]}$ of channels, i.e.,
\begin{equation}
\label{eq pi star}
    \pi^{*} = \argmin_{\pi \in \mathcal{D}(\mathcal{H}_{\pi})} \sum\limits_{t=1}^{T} \ell(\mathcal{E}^t, \pi).
\end{equation} For a sequence of channels $\mathcal{E}^{[T]}$, the performance of policy $f^{[T]}(\cdot)$ relative to the program state in $\pi^*$ (\ref{eq pi star}) is hence defined by the regret
\begin{equation}
\label{eq regret definition}
    \mathrm{Regret} (f^{[T]}, \mathcal{E}^{[T]}) = \sum\limits_{t=1}^{T} \ell(\mathcal{E}^t, \pi^t) - \sum\limits_{t=1}^{T} \ell(\mathcal{E}^t, \pi^*),
\end{equation}
with programs $\pi^t$ and $\pi^*$ given in (\ref{eq policy}) and (\ref{eq pi star}) respectively. 

To elaborate on the regret as the performance criterion of interest, observe first that, if the channel sequence $\mathcal{E}^{[T]}$ to be simulated were constant, i.e., if we had $\mathcal{E}^{t}=\mathcal{E}$ for some channel $\mathcal{E}$, obtaining a per-step regret $\ell(\mathcal{E}, \pi^t) - \ell(\mathcal{E}, \pi^*)$  that decreases with $t$ would indicate that the optimizer \eqref{eq policy} converges to the program that best approximates the channel in terms of the loss function $\ell(\mathcal{E},\pi)$. In the online setting under study, the channels $\mathcal{E}^{t}$ are allowed to vary arbitrarily, and the goal of the optimizer is to track such changes as they occur, i.e., as $t$ increases. Obtaining a small regret, irrespective of the channel sequence, provides evidence that the optimizer is extracting useful information about the single program $\pi^*$ that would have been optimal in hindsight. Specifically, following the standard online optimization framework \cite{Orabona_online_learning_notes}, we are interested in designing a policy $f^{[T]}$ that achieves a regret that grows sublinearly in $T$, implying that the per-step regret $\ell(\mathcal{E}^t, \pi^t) - \ell(\mathcal{E}^t, \pi^*)$ decreases over time $t$. 



\section{Matrix Exponentiated Gradient Descent}
\label{section MEGD and regret analysis}

In this section, we propose a policy for the problem of online channel simulation introduced in the previous section that is based on matrix exponentiated gradient descent (MEGD) \cite{MEGD_tsuda_2005}. We also analyze its regret, showing that it is sublinear in $T$.

    
\subsection{Matrix Exponentiated Gradient Descent (MEGD) for Online Channel Simulation}
\label{subsection MEGD for online channel simulation}

The proposed MEGD algorithm is an online mirror descent algorithm \cite{MEGD_tsuda_2005,Orabona_online_learning_notes} over the space of PSD matrices with unit trace. MEGD initializes the program state as the maximally mixed state $\pi^1 = I/2^{n_{\pi}}$. For every time $t\in [T]$ an arbitrary channel $\mathcal{E}^t$ is selected by nature, and the optimizer obtains the loss $\ell(\mathcal{E}^t, \pi^t)$. Based on the corresponding subgradient $g(\mathcal{E}^t, \pi^t)$, the optimizer updates the program state as
\begin{equation}
\label{eq MEGD update}
    \pi^{t+1} = \frac{\exp(Z^t)}{\mathrm{tr}(\exp(Z^t))},
\end{equation}
with matrix
\begin{equation}
    Z^t = \log (\pi^t) - \eta \Tilde{g}(\mathcal{E}^t, \pi^t),
\end{equation}
where $\Tilde{g}(\mathcal{E}^t, \pi^t) = \left(g(\mathcal{E}^t, \pi^t)+g(\mathcal{E}^t, \pi^t)^{\dagger}\right)/2$ represents the Hermitian part of the subgradient $g(\mathcal{E}^t, \pi^t)$ and $\eta >0$ is the learning rate.


    

To evaluate the subgradients $g(\mathcal{E}^t, \pi^t)$ for the losses (\ref{eq trace distrance}) and (\ref{eq infidelity metric}), we first define the quantum channel $\Lambda^t$ that maps program state $\pi^t$ to the corresponding Choi matrix $C_{\pi^t}$. This channel can be specified by its Kraus decomposition $\Lambda^t(\pi) = \sum_i A_i \pi A_i^{\dagger}$, where the $2^{2n} \times 2^n$ Kraus operators $\{A_i\}$ satisfy the condition $\sum_i A_i^{\dagger} A_i =I$, with $I$ being the $2^n \times 2^n$ identity matrix. The dual channel is defined as $\Lambda^t_{*}(\rho) = \sum_i A_i^{\dagger} \rho A_i$. Furthermore, we write the eigendecomposition of the Hermitian error operator $C_{\pi^t} - C_{\mathcal{E}^t}$ as $C_{\pi^t} - C_{\mathcal{E}^t} = \sum_{i} \lambda^t_i E^t_i$, with real eigenvalues $\{\lambda^t_i\}$ and eigenprojectors $\{E^t_i\}$. Following  \cite[Theorem $2$]{banchi_convex_opt_prog_q_comp_2020}, the subgradient with respect to PSD matrix $\pi$ of the loss functions $\ell_1 (\mathcal{E}^t,\pi)$ and $\ell_F (\mathcal{E}^t,\pi)$ evaluated at $\pi =\pi^t$ are given as
\begin{align}
    \label{eq gradient of trace distance}
    g_1(\mathcal{E}^t, \pi^t) &= \sum\limits_i \mathrm{sign}(\lambda^t_i) \Lambda_*^t(E^t_i),\\
    \label{eq gradient of infidelity metric}
    \text{and ~}g_F(\mathcal{E}^t, \pi^t) &= -\sqrt{1-\ell_F(\mathcal{E}^t, \pi^t)} \nabla L(\pi^t),
\end{align}
respectively, where $\mathrm{sign}(x)=1$ if $x\geq 0$ and $\mathrm{sign}(x)=-1$ if $x< 0$, and
\begin{align*}
    \nabla L(\pi^t) = \Lambda^t_* \left( \sqrt{C_{\mathcal{E}^t}} \left(\sqrt{C_{\mathcal{E}^t}} \Lambda^t(\pi^t) \sqrt{C_{\mathcal{E}^t}} \right)^{-\frac{1}{2}} \sqrt{C_{\mathcal{E}^t}} \right).
\end{align*}

Finally, for numerical stability, the MEGD update (\ref{eq MEGD update}) is implemented by replacing matrix $Z^t$ with the time-unrolled update \cite{MEGD_tsuda_2005}
\begin{equation}
\label{eq numerically stable MEGD update}
    Z^t = d^tI +\log (\pi^1) - \eta \sum_{\tau=1}^{t} \Tilde{g}(\mathcal{E}^{\tau}, \pi^{\tau}),
\end{equation}
where $\{d^t\}$ is a sequence of fixed constants. These constants do not affect an infinite-precision implementation of update (\ref{eq MEGD update}) with (\ref{eq numerically stable MEGD update}), but they can be useful to avoid numerical problems \cite{MEGD_tsuda_2005}. MEGD is summarized in Algorithm \ref{algorithm OLPQC}.


\begin{algorithm}[htbp]
\caption{Matrix exponentiated gradient descent (MEGD) for online channel simulation}
\label{algorithm OLPQC}
\begin{algorithmic}[1]
\State \textbf{Require:} 
Learning rate $\eta > 0$ and sequence of constants $d^t >0$ for $t\in [T]$.

\State Initialize the program state as $\pi^1 = I/2^{n_{\pi}}$

\For{$t\in [T]$}

\State Adversary selects a new channel $\mathcal{E}^t$

\State Optimizer obtains the loss $\ell(\mathcal{E}^t, \pi^t)$ using (\ref{eq trace distrance}) or (\ref{eq infidelity metric})

\State Optimizer computes subgradient $g(\mathcal{E}^t, \pi^t) $ of the loss $\ell(\mathcal{E}^t, \pi^t) $ at $\pi^t$ using (\ref{eq gradient of trace distance}) or (\ref{eq gradient of infidelity metric})

\State Optimizer updates the program state to $\pi^{t+1}$ using (\ref{eq MEGD update}) with (\ref{eq numerically stable MEGD update})
\EndFor

\end{algorithmic}
\end{algorithm}


\subsection{Regret Analysis}

To show that MEGD achieves sublinear regret in $T$, we start by interpreting the update rule \eqref{eq MEGD update} 
in terms of a regularized optimization problem that follows the mirror descent framework \cite{MEGD_tsuda_2005,Orabona_online_learning_notes,simeone2022machine}. To this end, we 
introduce the negative von Neumann entropy, defined as 
\begin{equation}
    \label{eq von Neumann entropy}
    F(\pi) = \mathrm{tr}(\pi \ln(\pi)) 
\end{equation}
for any state $\pi$. The Bregman divergence generated by function $F(\cdot)$ is given by   
\begin{align}
	B_F(\pi_1; \pi_2) &= F(\pi_1)-F(\pi_2) - \mathrm{tr}[\nabla F(\pi_2)(\pi_1-\pi_2)]= \nonumber\\ &=
\label{eq von Neumann divergence}
\mathrm{tr}(\pi_1 \ln(\pi_1) - \pi_1\ln (\pi_2)),
\end{align} 
and it corresponds to the quantum relative entropy \cite{Book_quantum_computing_quantum_information_Nielsen_chuang_2010}
between two density states $\pi_1$ and $\pi_2$, defined on the same Hilbert space. Following 
\cite{MEGD_tsuda_2005,Orabona_online_learning_notes}, the update rule \eqref{eq MEGD update} 
arises as the solution of the optimization problem
\begin{equation}
	\pi^{t+1} = \argmin_{\pi\in\mathcal D(\mathcal H_\pi)} \left(\eta~ \mathrm{tr}[\pi \tilde g(\mathcal E^t,\pi^t)]+ B_F(\pi;\pi^t) 
	\right),
\end{equation} where the first term is a linearization of the per-step loss and the second is a regularizer penalizing deviations from the current program $\pi^t$. 
To formulate the main result in Proposition \ref{proposition regret}
we also introduce  
the spectral norm $\| \cdot \|_{*}$ of a matrix $O$ as the square root of the largest eigenvalue of the matrix $O^{\dagger} O$.

\begin{proposition}
\label{proposition regret}

The regret (\ref{eq regret definition}) of MEGD is upper bounded as  
\begin{align}
\label{eq regret upper bound 1 in theorem statement}
    \mathrm{Regret} (f^{[T]}, \mathcal{E}^{[T]}) &\leq \frac{B_F(\pi^*; \pi^1)}{\eta} + \frac{\eta}{2 } \sum\limits_{t=1}^{T} \| g(\mathcal{E}^t, \pi^t)\|_{*}^2,
\end{align}
where subgradient $g(\mathcal{E}^t, \pi^t)$ and reference program state $\pi^*$ are defined in (\ref{eq gradient of trace distance})-(\ref{eq gradient of infidelity metric}) and (\ref{eq pi star}), respectively. Furthermore, if the spectral norm of the subgradient is bounded as $\| g(\mathcal{E}^t, \pi^t)\|_{*} \leq L_{*}$ for every $t \in [T]$, choosing the learning rate as $\eta = \sqrt{2\ln (2)n_{\pi}/T L_{*}^2}$, the following regret bound holds 
\begin{align}
\label{eq regret upper bound 2  in theorem statement}
    \mathrm{Regret} (f^{[T]}, \mathcal{E}^{[T]}) &\leq  L_{*} \sqrt{ 2\ln (2) n_{\pi} T}.
\end{align}
\end{proposition}

\begin{proof}
    Since the negative von Neumann entropy $F(\pi)$ is $1$-strongly convex over the space of density states $\pi \in \mathcal{D}(\mathcal{H}_{\pi})$ with respect to the trace norm \cite{strong_convexity_quantum_entropy_yu2013}, inequality (\ref{eq regret upper bound 1 in theorem statement}) is a consequence of \cite[Theorem $6.8$]{Orabona_online_learning_notes}. To prove (\ref{eq regret upper bound 2 in theorem statement}), we first upper bound $B_F (\pi^* ; \pi^1)$ as \begin{equation}
    \label{eq upper bound on B_F}
        B_F (\pi^* ; \pi^1) = F(\pi^*) + \ln (2^{n_{\pi}})  \leq \ln (2^{n_{\pi}}),
    \end{equation}
    where the inequality holds from the non-negativity of the von Neumann entropy $-F(\cdot)$ \cite{Book_quantum_computing_quantum_information_Nielsen_chuang_2010}. Using (\ref{eq upper bound on B_F}) in (\ref{eq regret upper bound 1 in theorem statement}), the regret can be upper bounded as
    \begin{equation}
    \label{eq regret upper bound 3}
        \mathrm{Regret} (f^{[T]}, \mathcal{E}^{[T]}) \leq \frac{\ln (2)n_{\pi}}{\eta} + \frac{\eta}{2} \sum\limits_{t=1}^{T} \| g(\mathcal{E}^t, \pi^t)\|_{*}^2.
    \end{equation}
Finally, inequality (\ref{eq regret upper bound 2  in theorem statement}) follows by using the  assumed inequality $\| g(\mathcal{E}^t, \pi^t)\|_{*} \leq L_{*}$ and by selecting the learning rate as indicated in the proposition.

\end{proof}


    

\section{Experiments}
\label{section experiments}

In this Section, we provide experimental results to validate the proposed MEGD scheme in Algorithm \ref{algorithm OLPQC}. We start by describing the generalized teleportation processor (GTP) which will be adopted as the programmable quantum processor $\mathcal{Q}$.

\subsection{Generalized Teleportation Processor (GTP)}
\label{subsection Generalized Teleportation Processor (GTP)}

\begin{figure}[htbp]
    \centering
    \includegraphics[height=1.9in]{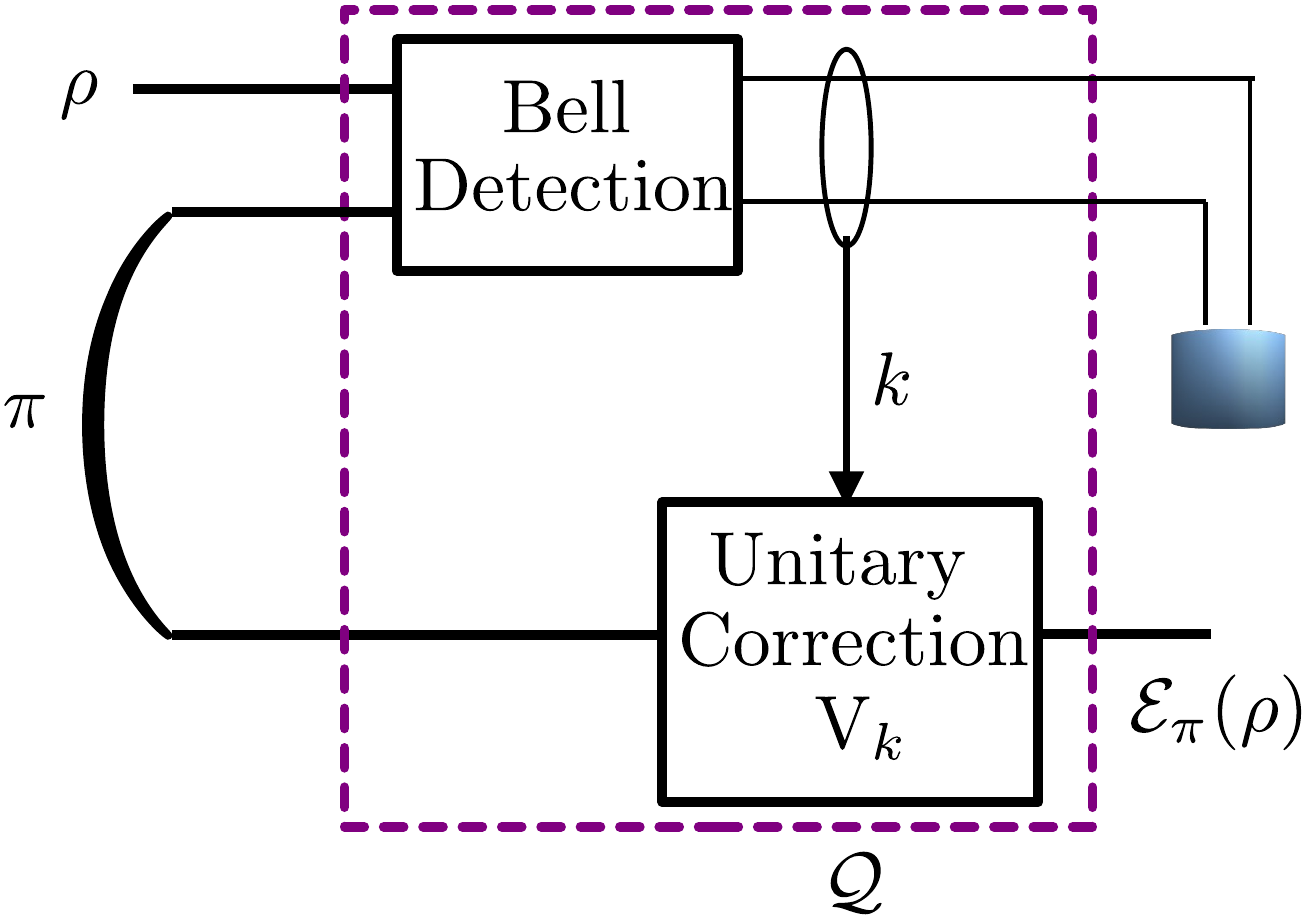} \vspace{-0.2cm}
    \caption{Generalized teleportation processor as a programmable processor $\mathcal{Q}$ operating on one input qubit $(n = 1)$ and on a two-qubit program state $\pi$ $(n_{\pi} = 2)$. 
		}
    \label{fig: teleportation processor}
    \vspace{-0.2cm}
\end{figure}


As illustrated in Fig. \ref{fig: teleportation processor}, the GTP operates on a register of three qubits with $n=1$ and  $n_{\pi} = 2$. A Bell measurement, defined by the projectors $\{P_0,P_1,P_2,P_3\}$, where

\footnotesize
\begin{equation*}
    P_0 = |\Phi^+\rangle \langle \Phi^+|, ~ P_1 = |\Psi^+\rangle \langle \Psi^+|, ~ P_2 = |\Psi^-\rangle \langle \Psi^-|, ~ P_3 = |\Phi^-\rangle \langle \Phi^-|,
\end{equation*} \normalsize
with $\{ |\Phi^+\rangle, |\Psi^+\rangle, |\Phi^-\rangle, |\Psi^-\rangle \}$ being the standard Bell states (see \cite[Section $1.3.6$]{Book_quantum_computing_quantum_information_Nielsen_chuang_2010}),
is applied to the input qubit and to the first control qubit. Then, depending on the output $k \in \{0,1,2,3\}$ of the measurement, with $k$ corresponding to projector $P_k$, a unitary correction $V_k$ is applied to the second control qubit, where 

\footnotesize
\begin{equation}
    V_0 = \begin{pmatrix}
    1 & 0\\
    0 & 1
    \end{pmatrix},  V_1 = \begin{pmatrix}
    0 & 1\\
    1 & 0
    \end{pmatrix},  V_2 = \begin{pmatrix}
    0 & -i\\
    i & 0
    \end{pmatrix},  V_3 = \begin{pmatrix}
    1 & 0\\
    0 & -1
    \end{pmatrix}
\end{equation} \normalsize
are the Pauli operators. The last qubit encodes the output state $\mathcal{E}_{\pi}(\rho)$ of the GTP as shown in Fig. 3.


It is known that the GTP can simulate exactly all teleportation-covariant channels, which include Pauli and erasure channels \cite{Fund_lim_of_repeaterless_Q_comm_pirandola_2017}. This is done by choosing the program state $\pi$ as the scaled Choi matrix $C_{\mathcal{E}}$ of the channel to be simulated. As a trivial special case, note that setting $\pi = |\Phi^+\rangle \langle \Phi^+|$ simulates the identity channel, since the corresponding Choi matrix is given by $ |\Phi^+ \rangle \langle \Phi^+|$. Not all channels are teleportation-covariant. For example, the amplitude damping channel is not teleportation-covariant \cite{Fund_lim_of_repeaterless_Q_comm_pirandola_2017}, and hence it cannot be exactly simulated using the GTP.


In order to implement Algorithm \ref{algorithm OLPQC} using the GTP, we need the quantum channel $\Lambda^t(\pi)$ mapping program state $\pi$ to the corresponding Choi matrix $C_{\pi}$ and its dual channel $\Lambda^t_* (\pi)$, for the evaluation of subgradients (\ref{eq gradient of trace distance})-(\ref{eq gradient of infidelity metric}). For the GTP, these are given by \cite{banchi_convex_opt_prog_q_comp_2020}
\begin{equation}
    \Lambda^t(\pi) = \Lambda^t_*(\pi) = \frac{1}{4} \sum\limits_{k=0}^{3} (V_k^{\dagger} \otimes V_k ) \pi^t (V_k^{\dagger} \otimes V_k )^{\dagger}.
\end{equation}

\subsection{Results}

We now study the simulation of the dephasing channel, which is a special case of Pauli channels. Specifically, at each time $t$, the single-qubit channel to be simulated is given as
\begin{equation}
\label{eq def of dephasing channel}
    \mathcal{E}^t(\rho) = (1-p^t)\rho + p^t Z \rho Z ,
\end{equation}
where $Z = V_3$ represents the Pauli $Z$ operator, and the sequence of probabilities $p^1, p^2, \ldots$ defines the sequence of channels (\ref{eq def of dephasing channel}).

In Fig. \ref{fig: uni_regret} we plot the normalized regret $T^{-1} \mathrm{Regret} (f^{[T]}, \mathcal{E}^{[T]})$,
where $\mathrm{Regret} (f^{[T]}, \mathcal{E}^{[T]})$ is defined in (\ref{eq regret definition}), as a function of time $T$, by considering the trace distance (\ref{eq trace distrance}) as the loss function. We consider $T \in [150]$ and use $\eta = 0.01$ as the learning rate and constants $d^t = 2$ for every $t \in [T]$ in (\ref{eq numerically stable MEGD update}). For each time window duration $T$, we obtain the  optimal constant reference program $\pi^*$ in (\ref{eq pi star}) by optimizing the sum of loss functions via the MEGD update rule in (\ref{eq MEGD update}) over 120 iterations with a learning rate $0.01/T$. We observed numerically that these choices yield convergent iterates. 
\footnote{The PyTorch code for regenerating the results of this paper is available at $<$https://github.com/kclip/OCOPQC$>$.}

\begin{figure}[htbp]
    \centering
    \includegraphics[height=2.6in]{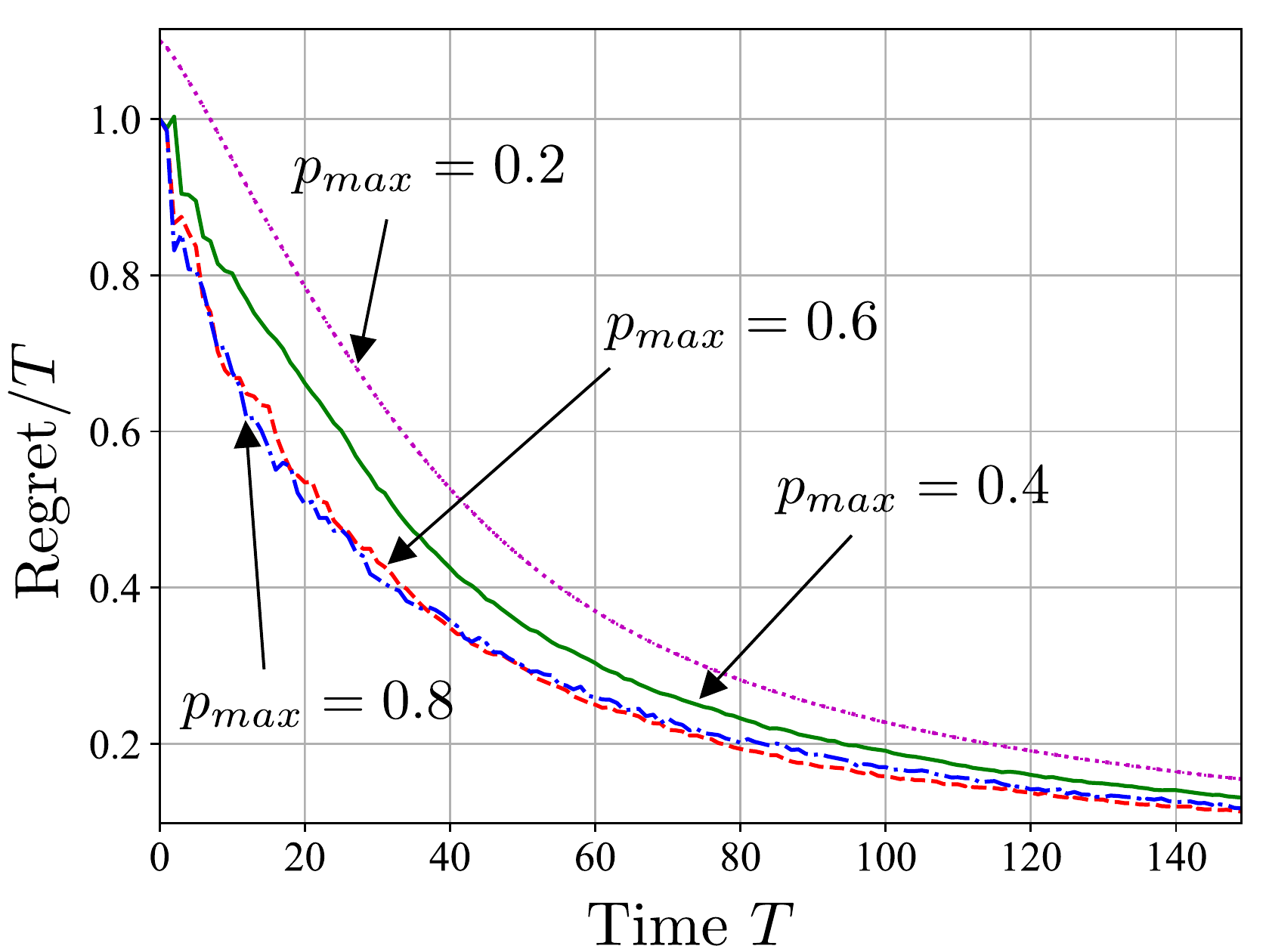}
    \caption{Normalized regret as a function of time $T$ for MEGD (Algorithm \ref{algorithm OLPQC}) when simulating a time-varying dephasing channel with dephasing probabilities drawn independently and uniformly at each time in the interval $[0.2,p_{max})$ (setting $p_{max} = 0.2$ models a constant channel).}
    \label{fig: uni_regret}
\end{figure}


We consider the setting in which the channel to be simulated changes independently at each time $t \in [T]$, with probability $p^t$ drawn uniformly in the interval $[0.2, p_{max})$, with $p_{max} \in \{0.2,0.4,0.6,0.8\}$.
In all cases, we observe that, as stated in Proposition \ref{proposition regret}, MEGD is able to obtain a normalized regret that decreases sublinearly with $T$, hence approaching the performance of the reference optimal constant program $\pi^*$. Furthermore, as $p_{max}$
increases, the performance of the fixed optimum program $\pi^*$ decreases and the regret of MEGD is reduced accordingly.

\bibliographystyle{IEEEtran}
\bibliography{IEEEabrv,cite.bib}

\end{document}